\documentclass[aps,prl,twocolumn,showpacs,superscriptaddress,groupedaddress]{revtex4-2}
\usepackage[utf8]{inputenc}
\usepackage{amsmath, amsthm, amssymb, amsfonts}
\usepackage[makeroom]{cancel}
\usepackage{graphicx}
\usepackage{bm}
\usepackage{dsfont} 
\usepackage{mathtools}
\usepackage{booktabs}
\usepackage{thm-restate}
\usepackage[colorlinks=true, linkcolor=blue, urlcolor=blue,citecolor=blue]{hyperref}
\usepackage{orcidlink}

\definecolor{Green1}{rgb}{0.0, 0.50, 0.0}
\definecolor{Graphgreen}{rgb}{0.0, 0.85, 0.0}
\definecolor{GFgreen}{rgb}{0.0, 1.0, 0.0}
\definecolor{Shadowred}{rgb}{0.70, 0.0, 0.0}
\definecolor{Scottred}{rgb}{0.85, 0.0, 0.0}
\definecolor{Huberred}{rgb}{1.0, 0.0, 0.0}

\DeclareMathOperator{\sign}{sign}
\DeclareMathOperator{\imm}{imm}
\DeclareMathOperator{\per}{per}
\DeclareMathOperator{\Wg}{Wg}
\DeclareMathOperator{\wg}{wg}

\DeclareMathOperator{\tr}{tr}
\DeclareMathOperator{\SEP}{SEP}
\DeclareMathOperator{\sep}{sep}
\DeclareMathOperator{\ent}{ent}
\DeclareMathOperator{\rank}{rank}
\DeclareMathOperator{\id}{id}

\newcommand{\GHZ}{{\rm{GHZ}}}
\newcommand{\bra}[1]{\mathinner{\langle #1|}}
\newcommand{\ket}[1]{\mathinner{|#1\rangle}}

\newcommand{\dyad}[1]{| #1\rangle \langle #1|}
\newcommand{\ot}[0]{\otimes}
\newcommand{\one}[0]{\mathds{1}}

\newcommand{\vv}{\ket{v_1} \ot \dots \ot \ket{v_n}}

\newcommand{\cdn}{(\C^d)^{\otimes n}}

\newcommand{\C}{\mathds{C}}
\newcommand{\WW}{\mathcal{W}}

\DeclareUnicodeCharacter{202C}{\^{i}}

\makeatletter
\newcommand{\vast}{\bBigg@{4}}
\newcommand{\Vast}{\bBigg@{5}}
\makeatother
\newtheorem{theorem}    {Theorem}
\newtheorem{proposition}[theorem]{Proposition}
\newtheorem{observation}[theorem]{Observation}
\newtheorem{lemma}      [theorem]{Lemma}

\usepackage{soul,xcolor}

\begin{document}

\setstcolor{red}
\title{
Entanglement detection with trace polynomials
}
\author{Albert Rico${}^{\orcidlink{0000-0001-8211-499X}}$}
\affiliation{
Faculty of Physics, Astronomy and Applied Computer Science, Institute of Theoretical Physics, Jagiellonian University,
30-348 Krak\'{o}w, 
Poland}
\author{Felix Huber${}^{\orcidlink{0000-0002-3856-4018}}$}
\affiliation{
Faculty of Physics, Astronomy and Applied Computer Science, Institute of Theoretical Physics, Jagiellonian University,
30-348 Krak\'{o}w, 
Poland}
\date{\today}
\begin{abstract}
We provide a systematic method for nonlinear entanglement detection based on trace polynomial inequalities. 
In particular, this allows to employ multi-partite witnesses for the detection of bipartite states, and vice versa. We identify pairs of entangled states and witnesses for which linear
detection fails, but for which nonlinear detection succeeds.
With the trace polynomial formulation a great variety of witnesses arise from immanant inequalities,
which can be implemented in the laboratory through the randomized measurements toolbox.
\end{abstract}

\maketitle
The experimental detection of entanglement is an ongoing challenge~\cite{Wang_DetectNearTermDevices,MorelliImprMeas_2022,Frerot_MBodyDetect2022},
for which a key tool are {entanglement witnesses}~\cite{Guhne2009EntDet}. 
These detect some entangled states by virtue of having a negative expectation value,
separating them from the set of fully separable states, i.e. from convex combinations of product states.
While every entangled state can be detected by some witness, the construction of witnesses is not a straightforward task~\cite{Guhne2009EntDet}
and frequently relies on making use of specific structure in the state to be detected~\cite{Guhne_DetectGraph,Toth_DetectStabilizer2005,Dagmar_DetectHypergraph2017,Huber2022DimFree,Jafarizadeh_DetectBellDiag2005,Duan_DetectContVar2000,Simon_DetectContVar2000,Marconi_EntSymmetric2021}. 

{\em Nonlinear} entanglement detection has become a recent focus of attention due to the development of the randomized measurement toolbox~\cite{Elben22Toolbox},
making local unitary invariants like partial transpose moments experimentally accessible through
single-copy measurements~\cite{Elben2020,Neven_SymmetryResMomentsPT2021}.
However, the currently available techniques are limited and a systematic development of nonlinear witnesses is desireable.
The aim of this paper is to provide such a systematic method that is not only suitable for the randomized measurement framework,
but also makes a broader use of known constructions.
The basic task we study is the following:
{\it given an entanglement witness $W$,
can it be employed also in a nonlinear fashion as to detect entanglement in multiple copies  $\varrho^{\otimes k}$?}

Here we answer this question in the affirmative.
In particular we show that:
\noindent i) having access to multiple copies of the state 
it is possible to detect entanglement locally,
where the size of states and witnesses can be different
[Observation~\ref{obs:23} and Figs.~\ref{fig:AliceBob2copies} and \ref{fig:DetIso}];
ii) that
there exist pairs of states and witnesses for which linear entanglement detection fails, but for which nonlinear detection succeeds [Observation~\ref{obs:IsotropicConcentration}]; and iii) that there is a large class of nonlinear witnesses arising from  trace polynomial inequalities [Fig.~\ref{fig:WernerDetProjections}].
These can analytically be treated in the group ring $\C S_n$ and are experimentally accessible through randomized measurements~\cite{Elben22Toolbox}.

\smallskip

\begin{figure}[tbp]
\includegraphics[]{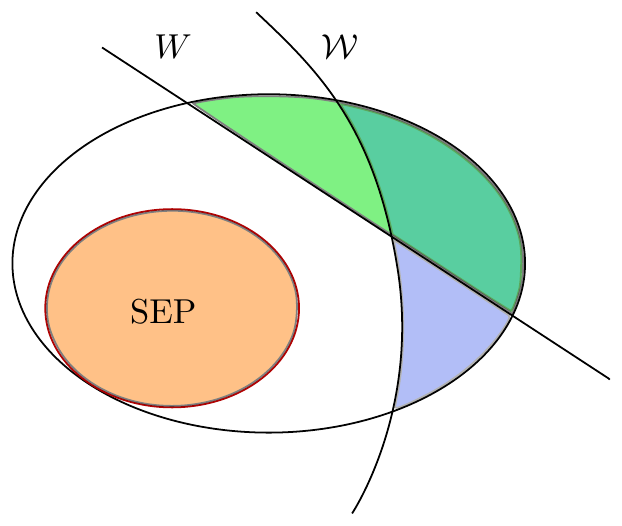}
    \caption{{\bf Linear versus nonlinear entanglement detection.} A linear witness $W$ defines a hyperplane in the state space, separating some entangled states (green area) from the rest; a nonlinear witness $\WW$ cuts the state space analogously in a nonlinear hypersurface, thus detecting a different set of entangled states (blue area).}
\label{fig:NonlinearCut}
\end{figure}

{\it Nonlinear entanglement witnesses.~---}
A quantum state is called separable if it can be written as convex combination of product states,
\begin{equation}
\varrho = \sum_i p_i \varrho_i^{(1)} \ot \varrho_i^{(2)}\ot \dots \ot \varrho_i^{(n)}\,,
\end{equation}
and entangled otherwise.
Entanglement detection with a witness works in the following way: 
to detect an entangled state $\varrho_{\ent}$,
a witness is an observable $W$ such that $\tr(W\varrho_{\ent}) < 0$ holds, while
$\tr(W\varrho_{\sep}) \geq 0$
for all separable states $\varrho_{\sep}\in\SEP$.
In this way, $W$ acts as a hyperplane, 
separating  a subset of entangled states from the rest.
We call this {\it linear detection}, in contrast to nonlinear detection introduced below.
This is illustrated in Fig.~\ref{fig:NonlinearCut}.

How can one find {\em nonlinear} witnesses $\WW$?
A range of methods are based on spin-squeezing inequalities and purity inequalities~\cite{HoroMulticopyWit2003,
KotowskiUniversalNonlinWit_2010,
Gessner_NonlinSpinSqueezing2019,
Trenyi_MulticopyMetrology2022}, nonlinear corrections to linear witnesses~\cite{Guhne_NonlinearImproveLinear_2006}, 
and multicopy scenarios show surprising entanglement activation properties~\cite{Yamasaki2022activationofgenuine}.
In a multi-copy scenario, one asks that a witness satisfies
$\tr(\WW\varrho_{\ent}^{\ot k}) < 0$
, while 
$  \tr\big (\WW\varrho_{\sep}^{\ot k}\big ) \geq 0
$
for all separable states $\varrho_{\sep}\in\SEP$.
Setting $k=1$ then recovers the standard use of witnesses.

Our approach here is to take a tensor product of linear witnesses,
\begin{equation}\label{eq:NonlinearWitness}
    \WW = W_1 \ot \dots \ot W_n\,.
\end{equation}
Naturally, the expectation values need to be computed in a manner such that the dimensions of $\WW$ and $\varrho^{\ot k}$ match (see Fig.~\ref{fig:RowColNotation} for a detailed explanation).
For example, for a tripartite state $\varrho_{ABC}$, take the tensor product of three bipartite witnesses,
\begin{equation}\label{eq:example_eval}
 \WW = U_{AA'} \ot V_{BB'} \ot W_{CC'}\,.
\end{equation}
Then, the expression
\begin{equation}\label{eq:example}
 \langle \WW \rangle_{\varrho^{\ot 2}}    = \tr\big(\WW (\varrho_{ABC} \ot \varrho_{A'B'C'} ) \big) \,
\end{equation}
is non-negative if $\varrho$ is separable.
\begin{figure}[tbp]
    \centering
    \includegraphics[scale=0.96]{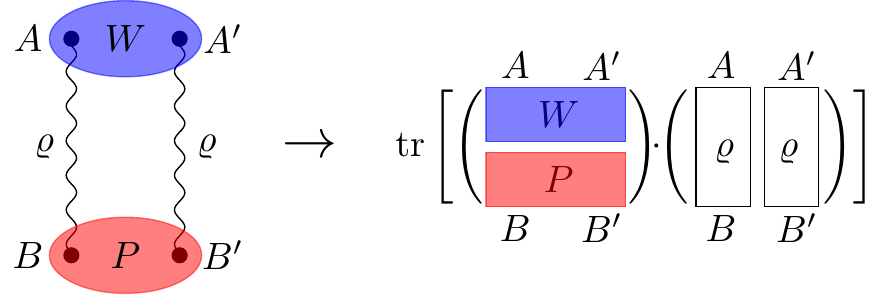}
    \caption{{\bf Sketch of the entanglement concentration scheme.} Expressions of the type $\langle W\ot P\rangle_{\varrho^{\ot 2}}$ can be obtained by first measuring two copies of Bob's subsystems (in red), and then two copies of Alice's subsystems (in blue).
    When $\varrho = \dyad{\psi}$ is a pure state such that $\ket{\psi}=\one\ot S\ket{\phi^+}$ with $S$ invertible, Bob can teleport his part of $\varrho$ to Alice by using $P=\dyad{\varphi}$ with $\ket{\varphi}=\one\ot {S^{\dag}}^{-1}\ket{\phi^+}$. Then standard linear witness evaluation $\tr\big (W\dyad{\psi}\big )$ is obtained as a particular case of the nonlinear method proposed in this letter (Appendix~\ref{app:LinearFromTelep}).}
    \label{fig:AliceBob2copies}
\end{figure}
At first sight, it may not be clear why Eq.~\eqref{eq:example} can detect entanglement, since the witnesses act along the cuts $A|A'$, $B|B'$ and $C|C'$ and $\varrho_{ABC}^{\ot 2}$ is separable in the cut $ABC|A'B'C'$. This apparent contradiction is resolved by realizing that the tensor product of witnesses for $A|A'$, $B|B'$, $C|C'$ is not necessarily a witness for $ABC|A'B'C'$.
The following shows that this method can indeed work:
\begin{observation}\label{obs:23}
Bipartite witnesses can be used to detect multi-partite entangled states nonlinearly.
\end{observation}
To see this take the $2$-qubit witnesses~\cite{Hyllus_2005QWitnessesBellIn}
\begin{equation}\label{eq:WitnessesBell}
\begin{aligned}
    W &= \one - X\ot X - Z\ot Z\,, \quad &
    V &= \dyad{\phi^+}^{\Gamma}\,,
\end{aligned}
\end{equation}
where $\Gamma$ is the partial transpose and 
$
\ket{\phi^+}=(\ket{00}+\ket{11})/\sqrt{2}
$; 
and the Greenberger-Horne-Zeilinger state $\varphi = \dyad{\GHZ}$ with
$
\ket{\GHZ
}=(\ket{000}+\ket{111})/\sqrt{2}
$.
Then
\begin{equation}
\tr\big( (W_{AA'} \ot W_{BB'} \ot V_{CC'}) ({\varphi_{ABC}}^{\ot 2}) \big) 
=
-1/2\,.
\end{equation}
%
An example with $k=n$ is the detection of the Bell state $\ket{\psi^+}=(\ket{01}+\ket{10})/\sqrt{2}$, where the witnesses of Eq.~\eqref{eq:WitnessesBell} give $\langle W\ot V\rangle_{{\psi^+}^{\ot 2}}=-1/2$.
\smallskip

\noindent {\it Entanglement concentration ---}
\begin{figure}[tbp]
\centering
\includegraphics[]{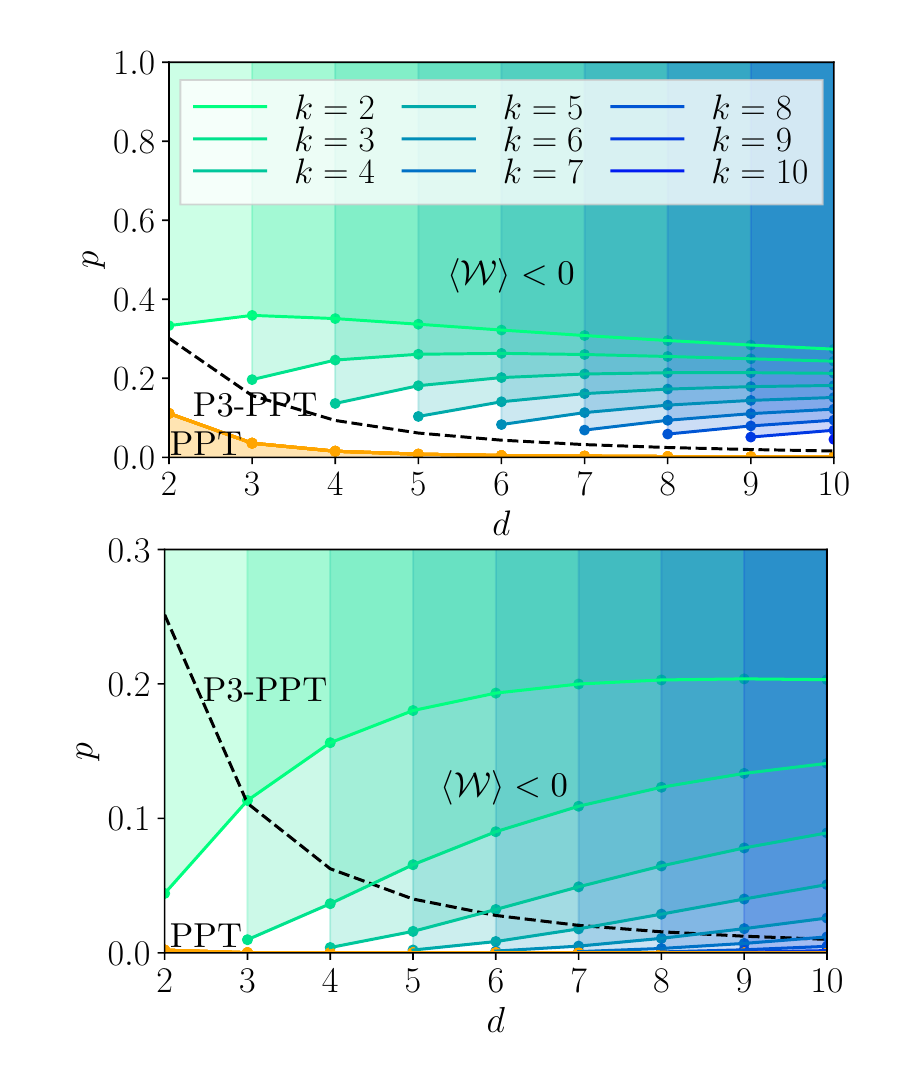}
\caption{{\bf Nonlinear detection of noisy 4- (top) and 10-partite (bottom) GHZ states.} Shown are the detection curves with $k=2,3,...,10$ copies of noisy n-partite GHZ states shared among an even number of parties $n$ with local dimension $d$ (top $n=4$; bottom $n=10$). We assume $d\geq k$ as otherwise the anti-symmetrizer vanishes.
All states with positive partial transpositions (PPT) are in the orange region, and hence states outside are entangled;
the green-blue regions are states that are detected by
Eq.~\eqref{eq:GHZ_W} with $W=\one-k!P_{1^k}$ and $\WW=W\ot P_{1^k}^{\ot n-1}$, even though $\tr(W\varrho)\geq 0$. Entangled states that are not detected by $\WW$ are in the white region, suggesting that detection gains robustness with the number of copies. The dashed line denotes the detection threshold by the P3-PPT criterion proposed in~\cite{Neven_SymmetryResMomentsPT2021}.
The evaluation of $\langle\WW\rangle$ and the P3-PPT criterion in terms of moments of the partial transpose
are detailed in Appendix~\ref{app:GHZdetection}.
\label{fig:DetIso}
}
\end{figure}
If in Eq.~\eqref{eq:NonlinearWitness} we replace some witnesses $W_i$ by positive semidefinite operators $P_i$, the expectation value with respect to $k$ copies of an $n$-qudit separable state remains non-negative. The following observation shows that this way, a witness with non-negative expectation value with respect to a certain entangled state can detect its entanglement.
\begin{observation}\label{obs:IsotropicConcentration}
There exist pairs of states $\varrho$ and witnesses $W$ for which
$\tr(W \varrho) \geq 0$, 
but for which there is $P\geq 0$
such that 
\begin{equation}
\tr\big ((W \ot P) \varrho ^{\ot k}\big ) < 0\,.
\end{equation}
In short: linear entanglement detection with $W$ fails, but nonlinear detection with $W \ot P$ succeeds.
\end{observation}
Consider the $n$-qudit Greenberger-Horne-Zeilinger state $\ket{\GHZ}= \tfrac{1}{\sqrt{d}}\sum_{i=0}^{d-1}\ket{i}^{\ot n}$ affected by white noise,
\begin{equation}\label{eq:noisyGHZgeneral}
    \varrho = p\dyad{\GHZ} + \frac{1-p}{d^n}\one\,,
\end{equation}
whose experimental generation and detection plays a central role in quantum computing and communication tasks~\cite{ChenGHZCriptography_2004,Bishop_GHZExp_2009,MooneyGHZExp_2021}.
Take the $n$-qudit witness
\begin{equation}\label{eq:Witness1-Pa}
    W = \one - n!P_{1^n}\,,
\end{equation}
where $P_{1^n}$ is the projector onto the fully antisymmetric subspace of $n$ qudits (given in Eq.~\eqref{eq:W1-P111} for $n=3$ and arbitrary dimension). 
The state $\varrho$ has a nonpositive partial transpose (and thus is entangled)
for $p > 1/(1+d^{n-1})$~\cite{GuhneGHZPPTqubits_2010,GabrielCritKsepMixedGHZ_2010}.

While $\tr(W\varrho)\geq 0$, given $\WW = (\one-k!P_{1^k})\ot P_{1^k}^{\ot n-1}$ one verifies that
\begin{equation}\label{eq:GHZ_W}
    \tr(\WW\varrho^{\ot k})<0
\end{equation}
for a range of values of $p$ (see Fig.~\ref{fig:DetIso}).
The evaluation of Eq.~\eqref{eq:GHZ_W} is given in Appendix~\ref{app:GHZdetection}.

With this prespective, one can generalize Eq.~\eqref{eq:NonlinearWitness} and evaluate $n$ observables $O_i\in L((\mathbb{\C}^d)^{\ot k})$ on $k$ copies of $\varrho\in L((\mathbb{\C}^d)^{\ot n})$ as
\begin{equation}\label{eq:non-lin_precise}
    \tr\left(\left( O_1\ot...\ot O_n\right) \varrho^{\ot k} \right)\,,
\end{equation}
thus certifying that $\varrho$ is entangled if the result is negative and the operators $O_i$ are either positive semidefinite operators or witnesses
(illustrated in Fig.~\ref{fig:RowColNotation}).
This way of combining positive operators and witnesses can be understood as follows: suppose Alice and Bob share two copies of an entangled state $\varrho_{AB}$ and $\varrho_{A'B'}$. If Bob measures $\one_{AA'}\ot P_{BB'}$, then with probability  
$
\tr(P\varrho_B\ot\varrho_{B'})
$, the state on Alice's side reads
\begin{equation}\label{eq:ConcAliceSide}
\xi_{AA'}=\frac{\tr_{BB'}\big ((\one_{AA'}\ot P_{BB'})(\varrho_{AB}\ot\varrho_{A'B'})\big )}{\tr(P\varrho_B\ot\varrho_{B'})}\,,
\end{equation}
where $\varrho_B=\tr_A(\varrho_{AB})$. 
If Alice measures 
$
\tr(W\xi) < 0
$ for some witness $W$, then $\varrho$ is entangled. 
In this case, Bob's measurement has concentrated the entanglement of $\varrho^{\ot 2}$ present in the partition $AA'|BB'$ into the subsystem $AA'$. This procedure is illustrated in Fig.~\ref{fig:AliceBob2copies}. 

For most pure states, this procedure of entanglement concentration on Alice's side can be understood in terms of entanglement swapping as follows. Suppose $\varrho$ is a pure state $\ket{\psi}=\one\ot S\ket{\phi^+}$ where the matrix $S$ is invertible with $\tr(SS^\dag) = 1$
and $\ket{\phi^+} = \sum_{i=0}^{d-1} \ket{ii} / \sqrt{d}$
~\footnote{This is the case if and only if the Schmidt rank of $\ket{\phi^+}$ and $\ket{\psi}$ are equal}. 
Then, by projecting $BB'$ onto the state $\ket{\varphi}=\one\ot {S^{\dag}}^{-1}\ket{\phi^+}$, Bob ($BB'$) effectively teleports his part of $\ket{\psi}$ to Alice ($AA'$) with nonzero probability using a second copy of the shared state $\ket{\psi}$ itself. When then Alice measures a witness $W$ the protocol reduces to standard linear detection,
evaluating 
$\tr(W\dyad{\psi})$. This is detailed in Appendix~\ref{app:LinearFromTelep}.

\smallskip
\noindent {\it Trace polynomial witnesses.~---}
How can one find suitable nonlinear witnesses that are simple to work with experimentally?
Here we focus on trace polynomials, due to the fact that these are experimentally accessible through the randomized measurement toolbox~\cite{Elben22Toolbox} (see Appendix~\ref{app:RandMes}).
A systematic construction of symmetric trace polynomials arises as follows:
matrix inequalities such as the Hadamard inequality,
\begin{equation}\label{eq:HadamardIneq}
    \prod_{i=1}^n A_{ii}\geq\det(A)
\end{equation}
for any $n\times n$ positive semidefinite matrix $A$, have been a long-standing topic of investigation~\cite{merris1997multilinear,GenMatFuncs1965}. 
In particular, the so-called  {\it immanant inequalities} generalize Eq.~\eqref{eq:HadamardIneq} and provide a large supply of $n$-partite linear witnesses of the form~\cite{MaassenSlides}
\begin{equation}
    W = \sum_{\sigma\in S_n}a_\sigma\eta_d(\sigma)\,,
\end{equation}
where $a_\sigma\in\C$ and $\eta_d(\sigma)$ is the representation of $\sigma\in S_n$ permuting the $n$ tensor factors of $(\mathbb{\C}^d)^{\ot n}$,
\begin{equation}
\eta_d(\sigma)\ket{i_1} \ot \dots\ot \ket{i_n}
=
\ket{i_{\sigma^{-1}(1)}} \ot \dots \ot \ket{i_{\sigma^{-1}(n)}}\,.
\end{equation}
How to transform matrix inequalities to witnesses in the symmetric group algebra over the complex numbers, $\C S_n = \{\sum_{\sigma \in S_n} a_\sigma \sigma \,|\, a_\sigma \in \C \}$, is sketched in Appendix~\ref{app:WitsFromImIneq}. For example, from Eq.~\eqref{eq:HadamardIneq} 
one obtains the $n$-qudit witness of Eq.~\eqref{eq:Witness1-Pa}. A few standard immanant inequalities and their corresponding witnesses are listed in Table~\ref{tab:ImmInWit}, and an extended list of nonlinear witnesses of different forms is given in Appendix~\ref{app:WitsFromImIneq}.

By recycling linear witnesses or positive operators $O_i\in\mathbb{\C}S_k$ arising from immanant inequalities and using them in the nonlinear way of Eq.~\eqref{eq:non-lin_precise}, the resulting nonlinear witness $\WW=O_1\ot...\ot O_n$ has interesting features.
Its expectation value $\langle\WW\rangle_{\varrho^{\ot k}}$ is an homogeneous polynomial of degree $k$ in $\varrho$,
acting on the copies of subsystems as {\it trace polynomials}~\cite{Huber2020PositiveMA,Klep_TrPolOptim2021}.
Since by the Schur-Weyl duality the action of the symmetric group $S_k$ commutes with that of
$k-$fold tensor products $X^{\ot k}$, the expression $\langle\WW\rangle_{\varrho^{\ot k}}$ is local unitary invariant.
Similarly to the recently introduced P3-PPT condition and other local unitary invariant quantities~\cite{Neven_SymmetryResMomentsPT2021,Elben2020},
trace polynomial witnesses can experimentally be measured using the randomized measurements approach.
There local projective measurements in random bases allow to estimate the expectation value of $k$-copy observables~\cite{Elben22Toolbox,ZhenhuanCorrelationLRMs_2022}.
This way, trace polynomial witnesses can be evaluated with classical post-processing by matrix multiplication,
thus avoiding storing large tensor products of matrices (see Appendix~\ref{app:RandMes}).

Insights into the strength of this approach can be obtained using Haar integration (Theorem 4.3 in~\cite{KuengHaar_2019}). This allows to compute the average expectation value of nonlinear trace polynomial witnesses with respect to Haar random states $\ket{\psi}\in\C^{d}$ with $d=d_1\cdots d_n$ over the Haar measure $\mathbf{d}_\mu$,
\begin{equation}\label{eq:HaarRandomState}
    \mathbb{E}[\langle\WW\rangle_{\psi}] = \int_{U\in\mathbf{U}(d)} \tr\left( \WW \, U\ket{\psi}\bra{\psi}U^{\dag} \right)\mathbf{d}_{\mu}(U).
\end{equation}
For example, entanglement concentration of a Haar-random $n$-partite state $\ket{\psi}$ with $n$ even and the witness $W=\one-k!P_{1^k}$ in Eq.~\eqref{eq:Witness1-Pa} yields a negative value on average (Appendix~\ref{app:HaarAverage}).

This approach also allows to design separability criteria based on the spectra of the state and its reductions.
For example, if a state $\varrho\in\C^d\ot\C^d$ with $\rank(\varrho)<d$ has a maximally mixed marginal, then it is entangled
(Observation~\ref{obs:CriteriaRank} in Appendix~\ref{App:Spectra}).

\begin{table}[tbp]
\begin{tabular}{ l  c  c  c  c  c  r  @{} l  c  c  c  c  c  r @{} l }
\toprule
{\bf Name} & & & & & &
{\bf Inequal} & {\bf ity} & & & & & &
{\bf Wit} & {\bf ness} \\
[4pt]
\midrule
\\
Hadamard & & & & & &
$\prod_iA_{ii}\,$ & $\geq\det(A)$ & & & & & &
$\one$ & $\,-\,n!P_{1^n}$ \\
[7pt]
Schur & & & & & &
$\frac{\imm_\lambda(A)}{\chi_\lambda(\id)}\,$ & $\geq\det(A)$ & & & & & &
$\frac{P_{\lambda}}{\chi_{\lambda}^2(\id)}$ & $\,-\,P_{1^n}$ \\
[7pt]
Hook & & & & & &
$\frac{\imm_{\lambda}(A)}{\chi_\lambda(\id)}\,$ & $\geq\frac{\imm_{\lambda'}(A)}{\chi_{\lambda'}(\id)}$ & & & & & &
$\frac{P_{\lambda}}{\chi_{\lambda}^2(\id)}$ & $\,-\,\frac{P_{\lambda'}}{\chi_{\lambda'}^2(\id)}$ \\
[7pt]
Marcus & & & & & &
$\per(A)\,$ & $\geq\prod_iA_{ii}$ & & & & & &
$n!P_{n}$ & $\,-\,\one$ \\
[7pt]
Permanent & & & & & &
$\per(A)\,$ & $\geq\frac{\imm_\lambda(A)}{\chi_\lambda(\id)}$ & & & & & &
$P_{n}$ & $\,-\,\frac{P_\lambda}{\chi_\lambda^2(\id)}$ \\
[7pt]
\bottomrule
\end{tabular}
\caption{{\bf Immanant inequalities and their corresponding witnesses.}
Listed are immanant inequalities and their corresponding entanglement witnesses.
Here  $\chi_\lambda$ is the character of the irreducible representation of the symmetric group $S_n$
labeled by the partition $\lambda \vdash n$; and $P_\lambda$ are Young projectors.
The Hook inequalities hold for hook tableaux with shape $\lambda = (j,1^{n-j})$
and then $\lambda' = (j-1,1^{n-j+1})$.
These concepts are further explained in Appendix~\ref{app:WitsFromImIneq}, with explicit examples for $n=3$.
The permanent dominance inequality is conjectured but has not yet been proven~\cite{merris1997multilinear}.\label{tab:ImmInWit}}
\end{table}

\smallskip
{\it Testing with Werner states.~---} The exponential growth of the Hilbert space with the number of subsystems limits our capability to numerically test this approach.
Using trace polynomial witnesses $W\in\C S_k$
allows to sample states and witnesses of nontrivial sizes with a symbolic algebra package,
which is computationally cheaper than computing the expectation value of exponentially large matrices.

\begin{figure}[tbp]
\centering
\includegraphics[scale=0.84]{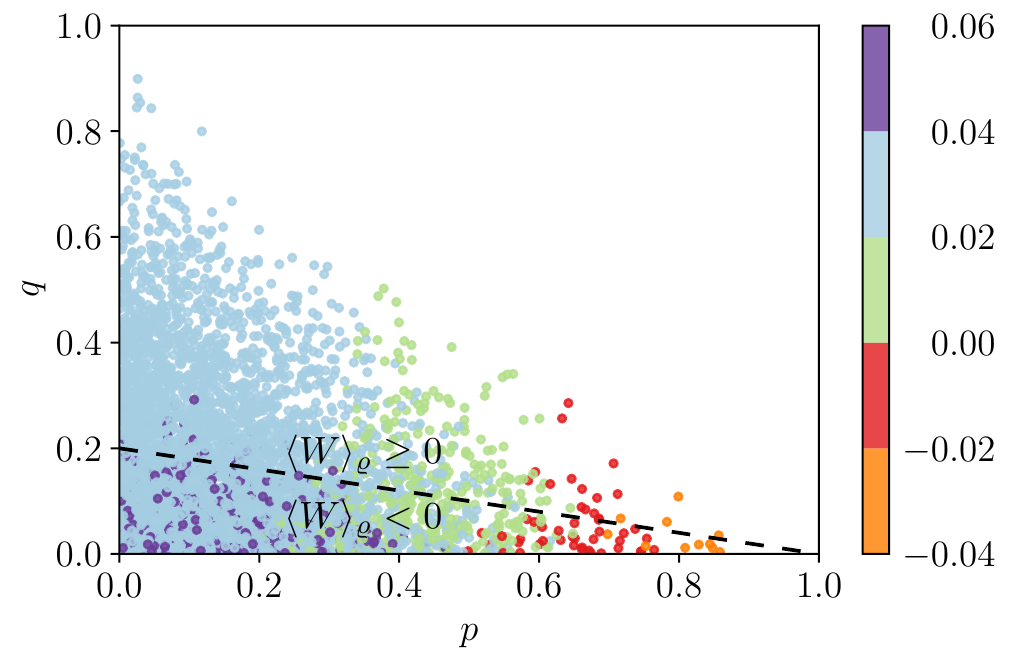}
\caption{{\bf Detection of 3 qutrit random Werner states using 3 copies.} We represent in colored dots a sample of 5000 random Werner states $\varrho$. The horizontal and vertical axes show their projections onto the antisymmetric and symmetric subspaces respectively, $p=\tr(\varrho P_{111})$ and $q=\tr(\varrho P_3)$, such that $\tr(\varrho P_{21})=1-p-q$. States below (above) the analytical dashed line are detected (not detected) linearly as $\langle W\rangle_{\varrho}<0$ ($\langle W\rangle_{\varrho}\geq 0$). The color gradient represents the expectation value $\langle\WW\rangle_{\varrho^{\ot 3}}$ with $\WW=W^{\ot 2}\ot P_{21}$, where $W=P_3-\tfrac{P_{21}}{4}$ is given in Eq.~\eqref{eq:WP3-P21} up to a global factor. This suggests that states with larger component in the antisymmetric subspace
are better detected by $\WW$.}\label{fig:WernerDetProjections}
\end{figure}

Sampling states and witnesses symbolically relies on group ring states and can be summarized as follows (for a detailed explanation see Appendix~\ref{app:WGFormalism}).
On the symmetric group $S_n$, a trace can be defined as
\begin{equation}
\tau(\sigma)=
\begin{cases}
n!\quad\text{if $\sigma=\id$}\,, \\
0\quad\text{else}\,.
\end{cases}
\end{equation}
Ref.~\cite{Huber2022DimFree} showed that given $r\in S_n$
with support only in $(\ker\eta_d)^\perp$
and $b\in\C S_n$,
it holds that
\begin{equation}\label{eq:Trace-Tau}
n!\tr\big (\Wg(d,n)\eta_d(r)\eta_d(b)\big ) = \tau(rb)\,,
\end{equation}
where $\Wg(d,n)$ is the Weingarten operator~\cite{Collins2006IntHaar,Procesi2020noteWg} defined in Eq.~\eqref{eq:defwg}. 
This equality allows us to compute in $\C S_n$ expectation values of the form~\eqref{eq:non-lin_precise} avoiding the exponential growth of the number of parameters with the local dimension $d$. 
Let us focus on the system size $k=n=d=3$. While a desktop computer cannot evaluate Eq.~\eqref{eq:non-lin_precise} in the Hilbert space $(\C^3)^{\ot 9}$, 
the fact that the nine-qutrit state factorizes as $\varrho^{\ot 3}$ makes this computation possible in $(\C S_3)^{\times 3}$ symbolically. For example, the expectation value of $\langle\WW\rangle_{\varrho^{\ot 3}}$ where $\WW=W^{\ot 2}\ot P_{21}$ with $W=P_3-\tfrac{P_{21}}{4}$ given in Eq.~\eqref{eq:WP3-P21}, evaluated on random Werner states is shown in  Fig.~\ref{fig:WernerDetProjections} (see Appendix~\ref{app:WGFormalism} for sampling details).
The results suggest that the detection of random states depends on their component within each irreducible subspace.

\smallskip
{\it Evaluating positive maps.~---}
Using variations of equation~\eqref{eq:non-lin_precise} one can detect nonpositive outcomes of positive maps applied locally. For example one can evaluate the reduction criteria~\cite{HoroRedCrit_1999} via purity conditions~\cite{HoroMulticopyWit2003} with Eqs.~\eqref{eq:Witness1-Pa} and~\eqref{eq:GHZ_W}, and detect nonpositivity of the partial transpose $\Gamma$ of a state $\varrho$ via
\begin{equation}\label{eq:NPTdetection}
    \tr\big ((\phi^+\otimes V)  (\varrho\otimes\sigma)\big ) = \tr\big (\varrho^\Gamma \sigma\big )
\end{equation}
if $\sigma$ lies in a negative eigenspace of $\varrho^{\Gamma}$.
In this sense, our method relates to the existence problem of nondecomposable tensor stable maps, which remain positive under tensor powers~\cite{MHermes_TensorStabMaps2015,Mirte_TensorStabMaps2022}, as follows: there exists a nondecomposable tensor stable map, if and only if, there exists a bipartite witness $W$ such that $W^{\ot n}$ cannot detect any pair of $n$-partite states $\varrho$ and $\sigma$. This is because a nondecomposable map $\Phi_W$ is tensor stable if and only if
\begin{equation}\label{eq:TensorStable}
    {\Phi_W}^{\ot n}(\rho)=\tr_1(W^{\ot n}(\rho^T\ot\one_d))\geq 0
\end{equation}
for any $n$-partite state $\varrho$, where $W$ is a witness for states with positive partial transpose (See Chapter 11 of~\cite{Karol-GeoQstates2006} and references therein); which is true if and only if the expectation value of ${\Phi_W}^{\ot n}(\rho)$ for any other $n$-partite state $\sigma$ is non-negative.

\smallskip
{\it Conclusions.~---}
Our approach shows how entanglement witnesses can be recycled in nonlinear fashion,
thus increasing the range of applicability of a given witness. In particular, this approach is suitable to randomized measurements~\cite{
Huang_ManyPropsFewMeas_2020,
Elben2020,WyderkaRandMeasCorr_2022}, 
trace polynomials being naturally invariant under local unitaries.

Some open questions remain:
How can one tailor nonlinear witnesses to specific states?
The Keyl-Werner theorem~\cite{Keyl-WernerEstimatingSpectra_2001} and subsequent work~\cite{Chris_Spectra_marg_2005} suggest that it is possible to choose witnesses according to the spectra of the reductions of the state in hand.
While Observation~\ref{obs:CriteriaRank} (Appendix~\ref{App:Spectra})
serves as a first step in this direction,
a systematic approach is yet missing.
Lastly, it would be interesting to make use of matrix inequalities involving elementary symmetric polynomials~\cite{merris1997multilinear},
which might be more powerful than the particular case of immanant inequalities.

\begin{acknowledgments}
We thank 
Ray Ganardi,
Otfried G\"uhne,
Pawe\l{} Horodecki,
Barbara Kraus,
Aaron Lauda, 
Anna Sanpera,
and 
Beno{\^{\i}}t Vermersch for fruitful discussions.
AR and FH are supported by the Foundation for
Polish Science through TEAM-NET (POIR.04.04.00-00-17C1/18-00).
\end{acknowledgments}

\bibliography{Bibliography}
\begin{widetext}
\appendix
\setcounter{secnumdepth}{1}

\section{Trace polynomial witnesses from randomized measurements}\label{app:RandMes}
Here we sketch how trace polynomial witnesses introduced in this work by recycling linear witnesses can be evaluated with randomized measurements, using the tools introduced in~\cite{Elben2020,Neven_SymmetryResMomentsPT2021,Elben22Toolbox}.

Let $\{W_i\}_{i=1}^n$ be $n$ witnesses for $k$-qudit mixed states, and let $\WW=\bigotimes_{i=1}^nW_i$ be a nonlinear witness.
The expectation value $\theta(\varrho)=\tr(\WW\varrho^{\ot k})$ can be approximated as follows:
First one applies random local unitaries $U_1\ot\dots\ot U_n$ to $\varrho$ and measures in the computational basis,
obtaining outcomes $q_1...q_n$.
Performing $m$ such measurements with a different choice of random local unitaries,
one obtains $m$ outcomes $q_1^{(r)}...q_n^{(r)}$ with $1\leq r\leq m$.
From this one constructs the classical shadows
\begin{align}
\tilde{\varrho}^{(r)}&=\bigotimes_{i=1}^n\tilde{\varrho}_i^{(r)} \quad\text{with}\label{eq:shadowr}\\
\tilde{\varrho}_i^{(r)}&=(d+1){U_i^{(r)}}^\dag \dyad{q_i^{(r)}}U_i^{(r)} - \one_d\,.\label{eq:shadowrfactors}
\end{align}
Then $\varrho^{(r)}$ has the same expectation value as $\varrho$ over the unitary ensemble and measurements~\cite{Huang_ManyPropsFewMeas_2020}.
By replacing each $j$'th copy of $\varrho$ by an operator $\tilde{\varrho}^{(r_j)}$, the expectation value $\tr(\WW\varrho^{\ot k})$ can be approximated as
\begin{equation}\label{eq:approxRMvalueShadows}
 \tilde{\theta}(\varrho)=\frac{1}{k!}{\binom{m}{k}}^{-1}\sum_{r_1\neq r_2\dots\neq r_k} \tr\bigg (\WW\varrho^{(r_1)}\ot\cdots\ot\varrho^{(r_k)}\bigg )
\end{equation}
with an error which decays as $1/\sqrt{m}$ for large $m$ (see~\cite{Elben2020,Neven_SymmetryResMomentsPT2021} and references therein). 

Expression~\eqref{eq:approxRMvalueShadows} can then be evaluated without computing the tensor product of classical shadows,
if $\WW=W_1\ot\dots\ot W_n$ is a trace polynomial witness with $W_i=\sum_{\sigma\in S_k}a_\sigma\eta_d(\sigma)$.
To see this, note that in this case Eq.~\eqref{eq:approxRMvalueShadows} decomposes as $\tilde{\theta}(\varrho)=\sum_{s}c_s\tilde{\theta}_s(\varrho)$ in terms of the form
\begin{equation}\label{eq:termTP}
 \tilde{\theta}_s(\varrho)=\tr\bigg (\bigotimes_{i=1}^n\eta_d(\sigma_i)\cdot\varrho^{(r_1)}\ot\dots\ot\varrho^{(r_k)}\bigg ),
\end{equation}
where $s=(\sigma_1,...,\sigma_n)$.
By recalling that each operator $\eta_d(\sigma_i)$ with $\sigma_i\in S_k$ permutes $k$ tensor factors in the party $i$ and that each shadow $\varrho^{(r_j)}$ is a product state in Eq.~\eqref{eq:shadowr}, one can now argue analogously as in Ref.~\cite{Elben2020}:
for any permutation $\sigma=(\alpha_1\dots\alpha_l)\dots(\xi_1\dots\xi_t)\in\C S_k$
and square matrices $X_1,\dots, X_k$, of equal size it holds that
\begin{equation}\label{eq:tprodtocontraction}
 \tr\big (\eta_d(\sigma) X_1\ot\dots\ot X_k\big )=\tr( X_{\alpha_1}\dots X_{\alpha_l} )\dots\tr(X_{\xi_1}\dots X_{\xi_{t}})\,.
\end{equation}
By assigning $X_j=\varrho_i^{(r_j)}$, the quantity $\tilde{\theta}(\varrho)$ can be evaluated as a matrix contraction of the tensor factors in Eq.~\eqref{eq:shadowrfactors} at each local party $i$, instead of computing a tensor product of $k$ shadows.

\section{Entanglement concentration with pure states of full Schmidt rank} \label{app:LinearFromTelep}
One can understand the entanglement concentration scheme as a variant of the entanglement swapping protocol. To see this, suppose Alice and Bob share two copies of a full Schmidt rank bipartite pure state $\ket{\psi}$. Since $\ket{\psi}$ is pure, it can be written as
\begin{equation}
\ket{\psi}=\one\ot S\ket{\phi^+} = S^T\ot\one\ket{\phi^+}\,
\end{equation}
where $S$ is a $d\times d$ complex matrix with $\tr(S^\dag S)=1$. By assumption $S$ is square and full rank, and therefore it has inverse $S^{-1}$. 

In this case, Bob can project onto $\dyad{\varphi}$ with $\ket{\varphi}=\one\ot {S^\dag}^{-1}\ket{\phi^+}$, thus obtaining the state $\xi$ in Eq.~\eqref{eq:ConcAliceSide} with certain probability $p$. It will be useful to consider the outcome state up to normalization, $\zeta=p\xi$, which reads
\begin{equation}
\zeta=\tr_{BB'}((\one_{AA'}\ot\dyad{\varphi}_{BB'})(\dyad{\psi}_{AB}\ot\dyad{\psi}_{A'B'}))\,.
\end{equation}
To compute $\zeta$ it will be convenient to use the row-column notation sketched in Fig.~\ref{fig:AliceBob2copies}, in which local parties are displayed in rows and copies are displayed in columns. The first factor inside the trace contains the measurement operators each party performs, and the second factor contains the copies of the state. This notation is further explained in Fig.~\ref{fig:RowColNotation}.
Then it becomes clear that $\zeta$ reads
\begin{align}
\zeta&=\frac{1}{d^3}\sum_{i,j,k,l,r,s=0}^{d-1}\tr_{BB'}\left(
\begin{pmatrix}
\one & \one \\
\ket{i}\bra{j} & {S^\dag}^{-1}\ket{i}\bra{j}S^{-1}
\end{pmatrix}
\begin{pmatrix}
S^T\ket{k}\bra{l}S^* & \ket{r}\bra{s} \\
\ket{k}\bra{l} & S\ket{r}\bra{s}S^\dag
\end{pmatrix}
\right)\\
&=\frac{1}{d^3}\sum_{i,j,k,l,r,s=0}^{d-1}
\begin{matrix}
S^T\ket{k}\bra{l}S^* & \ket{r}\bra{s} \\
\bra{j}k\rangle\bra{l}i\rangle & \bra{j}r\rangle\bra{s}i\rangle
\end{matrix}\\
&=\frac{1}{d^3}\sum_{r,s=0}^{d-1}S^T\ket{r}\bra{s}S^* \ot \ket{r}\bra{s}=\frac{1}{d^2}\dyad{\psi}\,,
\end{align}
which means that the final state on Alice's side is $\ket{\psi}$ with probability $p=1/d^2$ (given that Bob's outcome is $\ket{\varphi}$). This procedure is similar to the entanglement swapping protocol~\cite{Benett_Teleport1993,Jian-Wei_ExpSwapping1998}, except for the fact that we consider any pure quantum state with full Schmidt rank instead of restricting ourselves to maximally entangled states. When Alice obtains $\ket{\psi}$, she can detect its entanglement linearly with a witness $W_\psi$ fulfilling $\bra{\psi}W_{\psi}\ket{\psi}<0$. 

This can be summarized as follows: if $\ket{\psi}$ is full Schmidt rank and it can be detected linearly by a witness $W_\psi$, then it can also be detected using entanglement concentration with $\WW=\dyad{\varphi}\ot W_{\psi}$.

\section{Nonlinear entanglement dection of noisy GHZ states}\label{app:GHZdetection}
Consider the $n$-qudit GHZ state affected by white noise,
\begin{equation}
    \varrho = p\dyad{\GHZ} + \frac{1-p}{d^n}\one\,.
\end{equation}
The following result allows us to evaluate a nonlinear witness on $\varrho$ for arbitrary number of subsystems, copies and local dimensions. Some examples are illustrated in Figure~\ref{fig:DetIso}.
\begin{lemma}\label{lem:NonlinWitGHZ}
Let $\WW=(\one-k!P_{1^k})\ot P_{1^k}^{\ot n-1}$. The expectation value of $\WW$ with respect to $k$ copies of noisy $n$-partite GHZ state with local dimension $d$ reads
\begin{equation}\label{eq:appGHZdetGenN}
\tr\left(\WW\varrho^{\ot k}\right)=
\frac{
(d)_{(k)}}{k!^{n-1}}
\sum_{r=0}^k \binom{k}{r} \left(\frac{1-p}{d^n}\right)^r\left(\frac{p}{d}\right)^{k-r}
{(d+1-k)^{(r)}}^{n-2}\Delta_r
\end{equation}
with
\begin{equation}
    \Delta_r=d^r-\varsigma_n(k-r)!
(d+1-k)^{(r)}\,,
\end{equation}
where $(a)_{(b)} = a(a-1)\cdots(a+1-b)$ is the falling factorial and $(a)^{(b)} = a(a+1)\dots(a+b-1)$ is the rising factorial, and we define $\varsigma_n=(1+(-1)^n)/2$. 
\end{lemma}
\begin{proof}
First note that we can split
\begin{equation}
    \tr\left(\WW\varrho^{\ot k}\right) = T_1-k!T_2
\end{equation}
with
\begin{equation}
    T_1 = \tr\big (\one\ot P_{1^k}^{\ot n-1}\varrho^{\ot k}\big )\,,
\quad\quad
    T_2 = \tr\big (P_{1^k}^{\ot n}\varrho^{\ot k}\big )\,.
\end{equation}
Let us first compute $T_1$. To this end, it will be convenient to use the notation of Fig.~\ref{fig:AliceBob2copies} for tensor products mixing copies and subsystems.
We arrange operators at different positions such that elements on the same row act on different {\it copies} of a local party, e.g. $A_1$ and $A_2$; 
and operators in same column act on different {\it subsystems} of the same copy, e.g. $A_1$ and $B_1$. 
This is illustrated in~\ref{fig:RowColNotation}.
\begin{figure}[tbp]
    \centering
    \includegraphics[scale=0.124]{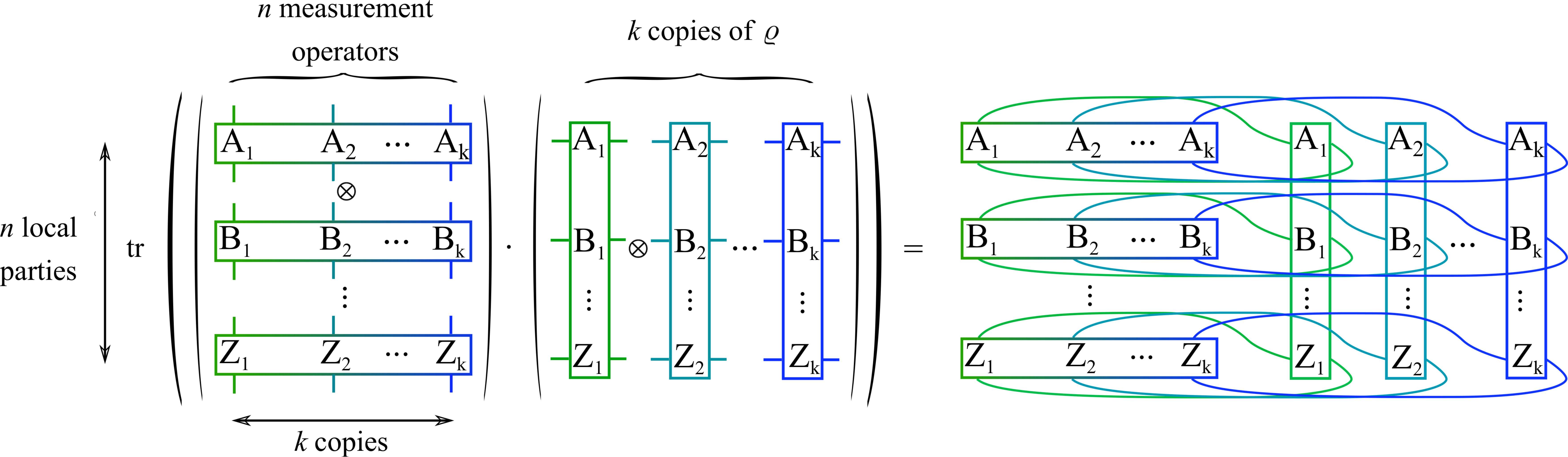}
    \caption{{\bf Graphical representation of the row-column notation for tensor products.} The full space is composed of $n$ local parties $A$, ..., $Z$, each of which is copied $k$ times. We denote the $i$'th copy of each party as $A_i$, ..., $Z_i$. In our expressions we compute the expectation values of $n$ operators (witnesses or projectors) acting across $k$ copies of each local subsystem, e.g. $A_1\otimes A_2\dots A_k$, with respect to $k$ operators ($k$ copies of a quantum state) acting on a single copy of all $n$ parties, e.g. $A_1\otimes B_1\cdots Z_1$. By assigning rows to local parties and columns to copies, this is represented by placing the operators in rows and columns correspondingly (left hand side). This representation can be understood as a tensor network notation of the expectation value to be evaluated by index contraction (right hand side).}
    \label{fig:RowColNotation}
\end{figure}
For example, 
\begin{equation}
\tr\left(
\begin{pmatrix}
X \\
Y 
\end{pmatrix}
\cdot
\begin{pmatrix}
M^{\ot 2}  \\
N^{\ot 2} 
\end{pmatrix}
\right)
= 
\tr\left(
\begin{pmatrix}
X \\
Y 
\end{pmatrix}
\cdot
\begin{pmatrix}
M  \\
N 
\end{pmatrix}^{\ot 2}
\right)
:=
\tr((X_{A_1A_2}\otimes Y_{B_1B_2})\cdot(M_{A_1}\ot M_{A_2}\ot N_{B_1}\ot N_{B_2})) \,.
\end{equation}
With this notation we write $T_1$ as
\begin{align}
T_1 = \tr \Vast (
\underbrace{\left(\begin{matrix}
\one_{d}^{\ot k} \\
P_{1^k} \\
\vdots \\
P_{1^k}
\end{matrix}\right)}_{(1)}
\cdot
{\underbrace{\left(
\frac{(1-p)}{d^n}
\begin{matrix}
\one_d \\
\one_d \\
\vdots \\
\one_d
\end{matrix}
+\frac{p}{d}\sum_{t,s=0}^{d-1}
\begin{matrix}
\ket{t}\bra{s} \\
\ket{t}\bra{s} \\
\vdots \\
\ket{t}\bra{s}
\end{matrix}
\right)}_{(2)}}^{\ot k}
\Vast )
\,. \\
\end{align}
Note when expanding above expression, 
each term from (2) contains 
$r$ copies of identity with a factor of $((1-p)/d^n)^r$,
and $k-r$ copies of the (unormalized) GHZ state with a factor 
of $(p/d)^{k-r}$.
Since both $\one_d^{\ot k} \ot P_{1^k}^{\ot n-1}$ in (1) are symmetric under permutation of the copies 
(i.e. permutations of the \emph{columns} in Fig.~\ref{fig:RowColNotation}), 
we obtain the standard binomial expression,
\begin{equation}\label{eq:row}
T_1 = \sum_{r=0}^k 
\binom{k}{r}\left(\frac{1-p}{d^n}\right)^r\left(\frac{p}{d}\right)^{k-r}\sum_{\Vec{t},\Vec{s}\in\mathbb{Z}_d^{k-r}} \tr \left(
\begin{matrix}
\one_{d}^{\ot k} \cdot \one_d^{\ot r}\ot\ket{\Vec{t}}\bra{\Vec{s}}\\
P_{1^k} \cdot \one_d^{\ot r}\ot\ket{\Vec{t}}\bra{\Vec{s}} \\
\vdots \\
P_{1^k} \cdot \one_d^{\ot r}\ot\ket{\Vec{t}}\bra{\Vec{s}}
\end{matrix}
\right)\,.
\end{equation}
Define
\begin{equation}
C_r=d^r\binom{k}{r}\left(\frac{1-p}{d^n}\right)^r\left(\frac{p}{d}\right)^{k-r}
\end{equation}
and note that the first row in Eq.~\eqref{eq:row} reads
\begin{equation}
\tr\Big (\one_{d}^{\ot k} \cdot \big ( \one_d^{\ot r}\ot\ket{\Vec{t}}\bra{\Vec{s}}\big )\Big ) =
\begin{cases}
d^r\quad\text{if}\,\Vec{t}=\Vec{s} \\
0\quad\,\,\,\text{if}\,\Vec{t}\neq\Vec{s}
\end{cases}\,.
\end{equation}
Then one has
\begin{align}
T_1 = 
\sum_{r=0}^k C_r \sum_{\Vec{t}\in\mathbb{Z}_d^{k-r}} \tr\Big( P_{1^k} \cdot (\one_d^{\ot r}\ot\ket{\Vec{t}}\bra{\Vec{t}}) \Big)^{n-1} \nonumber\\
= 
\sum_{r=0}^k C_r \sum_{\Vec{t}\in\mathbb{Z}_d^{k-r}} \tr\Big( \tr_{1\dots r}(P_{1^k}) 
\cdot 
\ket{\Vec{t}}\bra{\Vec{t}} \Big)^{n-1}\,.
\end{align}
The partial trace of $P_{1^k}$ with respect to any $r$ subsystems 
is known to be~\cite[Proposition 4]{Audenaert_Digest2006}
\begin{equation}
    \tr_{1...r}\big (P_{1^k}\big ) = \prod_{j=0}^{r-1}(d+1-k+j)\frac{(k-r)!}{k!}P_{1^{k-r}}\,,
\end{equation}
where $P_{1^{k-r}}$ is the antisymmetric projector acting on $k-r$ subsystems. Therefore we have
\begin{equation}
T_1 = \sum_{r=0}^k C_r \prod_{j=0}^{r-1}(d+1-k+j)^{n-1}
\left(\frac{(k-r)!}{k!}\right)^{n-1} 
\sum_{\Vec{t}\in\mathbb{Z}_d^{k-r}}  \tr\left(P_{1^{k-r}}\dyad{\Vec{t}}\right)^{n-1}\,.
\end{equation}
We will calculate separately the factor $\sum_{\Vec{t}\in\mathbb{Z}_d^{k-r}}  \tr\left(P_{1^{k-r}}\dyad{\Vec{t}}\right)^{n-1}$. Due to the properties of the antisymmetric subspace, one has $P_{1^{k-r}}\ket{\Vec{t}}=0$ unless all components of $\Vec{t} = (t_1, t_2, \dots, t_{k-r})$ are different. Therefore, $P_{1^{k-r}}\ket{\Vec{t}}$ is nonzero for
\begin{equation}
    d(d-1)(d-2)\cdots(d-(k-r+1)) = \prod_{j=1}^{k-r}(d+1-j)
\end{equation}
choices of $\Vec{t}$. In case it is non zero, we have
\begin{equation}
\bra{\Vec{t}}P_{1^{k-r}}\ket{\Vec{t}} = \frac{1}{(k-r)!}\sum_{\sigma\in S_{k-r}}(-1)^{\sign(\sigma)}
\bra{\Vec{t}}\sigma\ket{\Vec{t}}
= \frac{1}{(k-r)!}\,,
\end{equation}
where $\bra{\Vec{t}}\sigma\ket{\Vec{t}}$ vanishes if $\sigma\neq\id$. Therefore, we have
\begin{equation}
\sum_{\Vec{t}\in\mathbb{Z}_d^{k-r}}\tr\left(P_{1^{k-r}}\dyad{\Vec{t}}\right)^{n-1} = \frac{\prod_{j=1}^{k-r}(d+1-j)}{(k-r)!^{n-1}}\,.
\end{equation}
With this we are left with
\begin{equation}
T_1 = \sum_{r=0}^k \binom{k}{r} \left(\frac{1-p}{d^n}\right)^r\left(\frac{p}{d}\right)^{k-r}d^r\prod_{i=0}^{r-1}(d+1-k+i)^{n-1}\frac{(k-r)!^{n-1}}{k!^{n-1}} \frac{\prod_{j=1}^{k-r}(d+1-j)}{(k-r)!^{n-1}}\,.
\end{equation}
By shifting the indices, we have
\begin{equation}
    \prod_{i=0}^{r-1}(d+1-k+i)\prod_{j=1}^{k-r}(d+1-j) = \prod_{j=1}^{k}(d+1-j)\,
\end{equation}
and therefore $T_1$ reads
\begin{equation}
\begin{aligned}
T_1 &=  \frac{\prod_{i=1}^{k}(d+1-i)}{k!^{n-1}} 
\sum_{r=0}^k d^r\binom{k}{r} \left(\frac{1-p}{d^n}\right)^r\left(\frac{p}{d}\right)^{k-r}
\prod_{j=0}^{r-1}(d+1-k+j)^{n-2} \\
&= 
\frac{(d)_{(k)}}{k!^{n-1}} 
\sum_{r=0}^k d^r\binom{k}{r} \left(\frac{1-p}{d^n}\right)^r\left(\frac{p}{d}\right)^{k-r}{(d+1-k)^{(r)}}^{n-2}\,.
\end{aligned}
\end{equation}
Now we want to compute $T_2$. Following a similar procedure as the one described above, we have
\begin{equation}
\begin{aligned}
T_2 &= \tr \left(
\begin{matrix}
P_{1^k} \\
\vdots \\
P_{1^k}
\end{matrix}
\left(
\frac{(1-p)}{d^n}
\begin{matrix}
\one_d \\
\vdots \\
\one_d
\end{matrix}
+\frac{p}{d}\sum_{t,s=0}^{d-1}
\begin{matrix}
\ket{t}\bra{s} \\
\vdots \\
\ket{t}\bra{s}
\end{matrix}
\right)^{\ot k}
\right) \\
&= \sum_{r=0}^k \frac{C_r}{d^r}\sum_{\Vec{t},\Vec{s}\in\mathbb{Z}_d^{k-r}} \tr \left(
\begin{matrix}
P_{1^k} \cdot \one_d^{\ot r}\ot\ket{\Vec{t}}\bra{\Vec{s}} \\
\vdots \\
P_{1^k} \cdot \one_d^{\ot r}\ot\ket{\Vec{t}}\bra{\Vec{s}}
\end{matrix}
\right) \\
&= \sum_{r=0}^k \frac{C_r}{d^r} \sum_{\Vec{t},\Vec{s}\in\mathbb{Z}_d^{k-r}} \tr\left( P_{1^k} \cdot \one_d^{\ot r}\ot\ket{\Vec{t}}\bra{\Vec{s}} \right)\,.
\end{aligned}
\end{equation}
By noting that
\begin{equation}
    \tr\left( P_{1^k} \cdot \one_d^{\ot r}\ot\ket{\Vec{t}}\bra{\Vec{s}} \right) = \tr(\tr_{1...r}(P_{1^k})\ket{\Vec{t}}\bra{\Vec{s}})\,,
\end{equation}
we have
\begin{equation}\label{eq:AppT2trPn}
T_2 = \sum_{r=0}^k \frac{C_r}{d^r} \prod_{j=0}^{r-1}(d+1-k+j)^n\frac{(k-r)!^n}{k!^n} \sum_{\Vec{t},\Vec{s}\in\mathbb{Z}_d^{k-r}}  \tr\left(P_{1^{k-r}}\ket{\Vec{t}}\bra{\Vec{s}}\right)^n\,.
\end{equation}
For each pair of vectors $\Vec{t}$ and $\Vec{s}$, the factor $\tr\left(P_{1^{k-r}}\ket{\Vec{t}}\bra{\Vec{s}}\right)$ reads
\begin{align}\label{eq:ApptrPts}
\tr\left(P_{1^{k-r}}\ket{\Vec{t}}\bra{\Vec{s}}\right) 
&= \frac{1}{(k-r)!}\sum_{\sigma\in S_{k-r}}(-1)^{\sign(\sigma)}\bra{\vec{s}}\sigma\ket{\vec{t}}.
\end{align}
where $\sigma$ permutes the indices of the vectors 
$\Vec{s} \in \mathds{Z}_d^{k-r}$.
Recall that the components of $\Vec{s}$ are pairwise different, and the same holds for $\Vec{t}$. 
Consequently, Eq.~\eqref{eq:ApptrPts} is equal to $(-1)^{\sign(\sigma)}(k-r)!^{-1}$, 
where $\sigma \in S_{k-r}$ is the unique permutation such that $\sigma\ket{\Vec{t}}=\ket{\Vec{s}}$.

Thus to evaluate Eq.~\eqref{eq:AppT2trPn}, we need to sum Eq.~\eqref{eq:ApptrPts} over $\{\Vec{t}\in\mathbb{Z}_d^{k-r}\}$ with pairwise different components and over all the permutations of each vector $\Vec{t}$ of length $k-r$. Those two sets have cardinalities $\prod_{j=1}^{k-r}(d+1-j)$ and $(k-r)!$. Furthermore, note that $\sum_{\sigma\in S_{k-r}}\left((-1)^{\sign(\sigma)}\right)^n$ equals $(k-r)!$ if $n$ is even, and 0 if $n$ is odd. Those observations together lead to
\begin{equation}
\sum_{\Vec{t},\Vec{s}\in S_{k-r}}\tr\left(P_{1^{k-r}}\ket{\Vec{t}}\bra{\Vec{s}}\right)^n = \frac{\sum_{\Vec{t}\in\mathbb{Z}_d^{k-r},\sigma\in S_{k-r}}\left((-1)^{\sign(\sigma)}\right)^n}{(k-r)!^n}
= \varsigma_n \frac{\prod_{j=1}^{k-r}(d+1-j)}{(k-r)!^{n-1}}\,,
\end{equation}
where the function $\varsigma_n=(1+(-1)^n)/2$ vanishes if $n$ is odd. This allows us to compute $T_2$,
\begin{equation}
\begin{aligned}
T_2 &= \varsigma_n \sum_{r=0}^k \binom{k}{r} \left(\frac{1-p}{d^n}\right)^r\left(\frac{p}{d}\right)^{k-r}\prod_{i=0}^{r-1}(d+1-k+i)^n \prod_{j=1}^{k-r}(d+1-j) \frac{(k-r)!}{k!^n} \\
&= \varsigma_n\prod_{i=1}^{k}(d+1-i) \sum_{r=0}^k \binom{k}{r} \left(\frac{1-p}{d^n}\right)^r\left(\frac{p}{d}\right)^{k-r} \frac{(k-r)!}{k!^n}\prod_{j=0}^{r-1}(d+1-k+j)^{n-1} \\
&= \varsigma_n (d)_{(k)} \sum_{r=0}^k \binom{k}{r} \left(\frac{1-p}{d^n}\right)^r\left(\frac{p}{d}\right)^{k-r} \frac{(k-r)!}{k!^n}
{(d+1-k)^{(r)}}^{n-1}
\end{aligned}
\end{equation}
\end{proof}
We compare this result to the P3-PPT method proposed in~\cite{Neven_SymmetryResMomentsPT2021}, which consists of evaluating the LU-invariant expression
\begin{equation}
    \tau(\varrho)=\tr({\varrho^{T_A}}^3)-\tr({\varrho^{T_A}}^2)^2
\end{equation}
for each set $A \subset \{1,\dots, n\} $ of local parties.
By the symmetry of the GHZ and maximally mixed states, one observes that $\tau(\varrho)$ does not depend on the choice of the subsystem $A$.
One can evaluate the P3-PPT criterion shown in Fig.~\ref{fig:DetIso} by evaluating the second and third moments of the partial transpose,
\begin{equation}
\tr({\varrho^{T_A}}^2)=p^2+\frac{2p(1-p)}{d^n}+\frac{(1-p)^2}{d^n}
\end{equation}
and
\begin{equation}
\tr({\varrho^{T_A}}^3)=\frac{p^3}{d^2}+\frac{2p^2(1-p)}{d^n}+\frac{p(1-p)^2}{d^{2n}}+\frac{1-p}{d^n}\tr({\varrho^{T_A}}^2)\,.
\end{equation}

\section{Entanglement witnesses from immanant inequalities}\label{app:WitsFromImIneq}
We shortly sketch the connections established in Refs.~\cite{MaassenSlides,Huber2020PositiveMA,Huber2021MatrixFO}. Denote by $S_n$ the symmetric group permuting $n$ elements. Irreducible representations of $S_n$ are labeled by partitions $\lambda=(\lambda_1,...,\lambda_n)$ with integers $\lambda_1\geq\dots\geq\lambda_n$ satisfying $\lambda_1+...+\lambda_n=n$. The character associated to $\lambda$ is then the trace of the representation $\lambda$, denoted $\chi_\lambda(\sigma)=\tr(\lambda(\sigma))$.

Given an $n \times n$ matrix $A$,
one can then define the immanant associated to $\lambda$ as
\begin{equation}
    \imm_\lambda(A) = \sum_{\sigma \in S_n} \chi_\lambda(\sigma) \prod_{i=1}^n A_{i \sigma(i)}\,.
\end{equation}
For example, the symmetric and antisymmetric representations are labelled by $(n,0,...,0)=(n)$ and $(1,...,1)=(1^n)$, respectively. Their associated characters are $\chi_{(n)}=1$ and $\chi_{(1^n)}(\sigma)=\sign(\sigma)$. The corresponding matrix functions are the permanent and the determinant, 
\begin{align}
 \per(A)&= \imm_{(n)}(A) &
 \det(A)&= \imm_{(1^n)}(A)\,.
\end{align}
Now let $\eta_d(\sigma)$ be the representation of the symmetric group on $\cdn$ permuting its tensor factors as
\begin{equation}
    \eta_d(\sigma) \vv
    = 
    \ket{v_{\sigma^{-1}(1)}} 
    \ot \dots \ot 
    \ket{v_{\sigma^{-1}(n)}}\,.
\end{equation}
Maassen and K\"ummerer \cite{MaassenSlides} established the following relation between inequalities for immanants and entanglement witnesses.
\begin{proposition}{{\em \cite{MaassenSlides}}}\label{prop:MaassenWitness}
If   for all $A \geq 0$ it holds that
\begin{equation}\label{eq:DifferenceImmanants}
 a \imm_\lambda(A) - b \imm_\mu(A) \geq 0\,,
\end{equation}
then the operator
\begin{equation}
W = \frac{a}{\chi_\lambda(\id)}P_\lambda-\frac{b}{\chi_\mu(\id)}P_\mu
\end{equation}
has non-negative expectation value on the set of separable states, where $P_\lambda$ ($P_\mu$) is the projector onto the subspace corresponding to irreducible representation $\lambda$ ($\mu$).
\end{proposition}
\begin{proof}
Let $f(A)$ be a linear combination of permutations of entries of $A$ which is non-negative for $A\geq 0$, and consider $A$ as a Gram matrix of
vectors $\ket{v_1}, \dots, \ket{v_n}$.
Then 
\begin{equation}
\label{eq:wideeq}
\begin{split}
0\leq f(A) &= 
\sum_{\sigma\in S_n} c_\sigma \prod_{i=1}^n A_{i \sigma^{-1}(i)}  \\
&= 
\sum_{\sigma\in S_n} c_\sigma\prod_{i=1}^n \bra{v_i} v_{\sigma^{-1}(i)}\rangle  \\
&= 
\sum_{\sigma\in S_n} c_\sigma \tr \big (
\eta_d(\sigma)\bigotimes_{i=1}^n\ket{v_i}\bra{v_i}
\big )\,.
\end{split} 
\end{equation}
Proposition~\ref{prop:MaassenWitness} follows by assigning $f(A)=a \imm_\lambda(A) - b \imm_\mu(A)$  and considering the convex combination $\varrho=\sum_jp_j\bigotimes_{i=1}^n\ket{v_i^{(j)}}\bra{v_i^{(j)}}$, which can be seen as an unnormalized separable state. 
\end{proof}
For example, it holds that the permanent of a positive semidefinite matrix is larger than the determinant,
\begin{equation}\label{eq:pergeqdet}
 \per(A)\geq\det(A)\quad\text{for all}\,\,A\geq 0.
\end{equation}
From this inequality one obtains the witness $W=P_{(n)}-P_{1^n}$, where $P_{(n)}$ is the projector onto the fully symmetric subspace of $n$ qudits and $P_{(1^n)}$ is the rank-1 projector onto the fully antisymmetric subspace. Well-known immanant inequalities~\cite{merris1997multilinear} like Eq.~\eqref{eq:pergeqdet} and their associated witnesses are listed in Table~\ref{tab:ImmInWit}. While for $n=2$ the SWAP witness $\eta_d(12)$ is obtained, for $n=3$ and arbitrary dimension $d$ these inequalities provide the following independent witnesses:
\begin{align}
W_{\text{Hadamard}}&=\one-3!P_{1^3}=\eta_d(12)+\eta_d(13)+\eta_d(23)-\eta_d(123)-\eta_d(132) \label{eq:W1-P111} \\ 
W_{\text{Schur}}&=\frac{3!P_{21}}{4}-3!P_{1^3}=\eta_d(12)+\eta_d(13)+\eta_d(23)-\frac{3}{2}\Big (\eta_d(123)+\eta_d(132)\Big ) \label{eq:WP21-P111} \\ 
W_{\text{Hook}}&=3!P_{3}-\frac{3!P_{21}}{4}=\eta_d(12)+\eta_d(13)+\eta_d(23)+\frac{3}{2}\Big (\eta_d(123)+\eta_d(132)\Big ) \label{eq:WP3-P21} \\ 
W_{\text{Marcus}}&=3!P_{3}-\one=\eta_d(12)+\eta_d(13)+\eta_d(23)+\eta_d(123)+\eta_d(132) \label{eq:WP3-1}\,,
\end{align}
where we denote $\eta_d(123)\ket{0}\ot\ket{1}\ot\ket{2}=\ket{2}\ot\ket{0}\ot\ket{1}$ and similarly for other permutations.

From these witnesses acting on ${(\C^d)}^{\ot k}$, one can construct nonlinear witnesses of the form $\WW = W_1\ot...\ot W_n$, as in Eq.~\eqref{eq:NonlinearWitness}. However, one can also go out of this scheme
and use the same tools to design other $n$-qudit nonlinear witnesses by considering immanant inequalities involving matrices of different sizes, products of matrices and combinations of standard immanant inequalities~\cite{merris1997multilinear}. Below we sketch graphically some examples for the bipartite case. Given a partition $\lambda$ of a natural number $r$, we denote $\Tilde{P}_\lambda=r!P_{\lambda}/{\chi_{\lambda}}^2$ with $\chi_{\lambda}=\chi_{\lambda}(\id)$.
\begin{align}
&\,\,\text{{\bf Assumption}}
&\text{{\bf Matrix inequality}} \quad\quad
& 
&\text{{\bf Witness expression}}\quad\quad\quad 
&\phantom{,}\nonumber\\\nonumber\\
&A\geq 0\,,\quad B\geq 0
&\begin{aligned}
&\imm_\lambda(A)/\chi_\lambda\geq\imm_\gamma(A)/\chi_\gamma\\
&\imm_\mu(B)/\chi_\mu\geq\imm_\nu(B)/\chi_\nu
\end{aligned}
&
&\vcenter{\hbox{\includegraphics[]{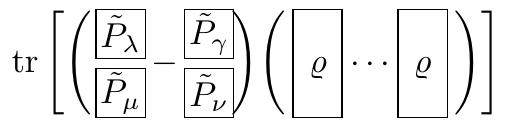}}}
\label{eq:AppImmsProd}
\\\nonumber\\
&\begin{pmatrix}
A_{11} & A_{12} \\
A_{12}^\dag & A_{22}
\end{pmatrix}\geq 0
&\per(A_{11})\per(A_{22})\leq\per(A)
&
&\vcenter{\hbox{\includegraphics[]{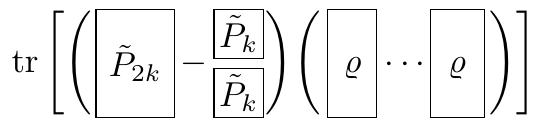}}}
\label{eq:AppBlockPer}
\\\nonumber\\
&A \geq 0\,,\quad B \geq 0
&\det(A\circ B)\geq\det(A)\prod_{i=1}^kB_{ii}
&
&\vcenter{
\hbox{\includegraphics[]{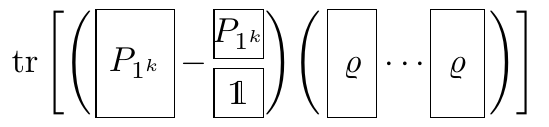}}}
\label{eq:AppHadDet} 
\end{align}
Here some notation needs to be clarified. In Eq.~\eqref{eq:AppImmsProd} $P_{\lambda}$, $P_{\mu}$, $P_{\gamma}$ and $P_{\nu}$ are the $k$-partite Young projectors corresponding to partitions of $k$ labelled by $\lambda$, $\mu$, $\gamma$ and $\nu$ respectively. In Eq.~\eqref{eq:AppBlockPer}, $P_{2k}$ is the $(2k)$-partite symmetrizer with local dimension $d$, which acts across all copies and all parties ($2k$ local subspaces in total). In Eq.~\eqref{eq:AppHadDet}, $P_{1^k}$ inside the large box is the antisymmetrizer of $k$ parties with local dimension $d^2$. That is, it reads
\begin{equation}
    P_{1^k} = \frac{1}{k!}\sum_{\sigma\in S_k} (-1)^{\sign(\sigma)}\sigma_A\ot\sigma_B\,,
\end{equation}
where each $\sigma_A$ permutes $k$ copies of the party $A$ and $\sigma_B$ permutes $k$ copies of the party $B$.

\section{Adapting witnesses to spectra}\label{App:Spectra}
The question of what states are better detected can be addressed using the fact that the global and local spectra of quantum states influence the inner product between Young projectors and tensor copies of the state at hand~\cite{Keyl-WernerEstimatingSpectra_2001,Chris_Spectra_marg_2005}. As an example, the following inequality involving a bipartite quantum state $\varrho_{AB}$ and two local projectors $P$ and~$Q$~\cite{Chris_Spectra_marg_2005},
\begin{equation}\label{eq:InequalityProjectors}
    \tr\big ( (P\otimes Q){} {\varrho_{AB}} ^{\ot k}\big )\geq\tr\big (P\varrho_A ^{\ot k}\big )+\tr\big (Q\varrho_B ^{\ot k}\big )-1\,,
\end{equation}
suggests that states with reduction $\varrho_A$ maximizing the overlap $\tr(P_{1^k}\varrho_A^{\ot k})$, also maximize the negative term in Eq.~\eqref{eq:AppHadDet}. 
In this direction, nonlinear witnesses from trace polynomials give rise to nonlinear separability criteria like the following one.
\begin{observation}\label{obs:CriteriaRank}
Let $\varrho$ be a bipartite state with local dimension $d$ such that $\rank(\varrho)<d$. The state $\varrho$ is detected by Eq.~\eqref{eq:AppHadDet} if it has at least one maximally mixed local reduction $\varrho_A\propto\one$.
\end{observation}
\begin{proof}
Let us start from the witness of Eq.~\eqref{eq:AppHadDet},
\begin{equation}
    \WW = P^{(d^2)}_{1^k}-P^{(d)}_{1^k}\ot\one_d^{\ot k}\,.
\end{equation}
Here $P^{(d^2)}_{1^k}$ is the antisymmetrizer acting on $k$ copies of the full bipartite Hilbert space $\C^d\ot\C^d$; $P^{(d)}_{1^k}$ is the antisymmetrizer acting on $k$ copies of the Hilbert space $\C^d$ of the party $A$; and $\one_d^{\ot k}$ is the identity operator acting on $k$ copies of the Hilbert space $\C^d$ of the party $B$. We will restrict ourselves to $k\leq d$ copies, as otherwise $P^{(d)}_{1^k}=0$ and therefore $\WW$ is positive semidefinite.

Now let $\varrho$ have spectrum $\{r_1,...,r_{d^2}\}$ ordered by $r_1\geq r_2\geq...\geq r_{d^2}$. As a particular case of Theorem 1 in~\cite{Chris_Spectra_marg_2005} we have the following bound,
\begin{equation}\label{eq:BoundChristandl}
\tr\big (P_{1^k}^{(d^2)}\varrho ^{\ot k}\big )\leq (k+1)^{d^2(d^2-1)/2} k\prod_{i=1}^k r_i\,,
\end{equation}
where $\{r_i\}$ are the largest $k$ eigenvalues of $\varrho$. By assumption, $\rank\varrho<d$ and consider a number of copies between $\rank\varrho$ and $d$, $\rank\varrho<k\leq d$. Then there exists at least a zero eigenvalue $r_{i^{*}}=0$ with $i^{*}\leq k$ and the non-negative term of $\langle\WW\rangle_{\varrho^{\ot k}}$ vanishes,
\begin{equation}
\tr\big (P_{1^k}^{(d^2)}{} \varrho ^{\ot k}\big ) = 0\,.
\end{equation}
In this case, the expectation value of $\WW$ is either 0 if $\tr\big (P_{1^k}^{d}\varrho_A ^{\ot k}\big )=0$, or strictly negative if $\tr\big (P_{1^k}^{d}\varrho_A ^{\ot k}\big )>0$. By Corollary 2 in~\cite{Chris_Spectra_marg_2005}, if $\varrho_A$ is maximally mixed, then there exists a number of copies $k_0\leq d$ such that for all $k\geq k_0$, its projection onto the antisymmetric subspace is nonzero,
\begin{equation}
    \tr(P_{1^k}\varrho_A^{\ot k}) > 0\,.
\end{equation}
Therefore the expectation value of $\WW$ is strictly negative, which means that $\varrho$ is entangled.
\end{proof}
\noindent Remark: That such state is entangled can also be seen by a variant of the range criterion~\cite{Horo1997PPT2x4}.
\section{Average expectation value over Haar random states}\label{app:HaarAverage}
Here we will apply the entanglement concentration technique to Haar random states. We will evaluate the average expectation value of the nonlinar witness $\WW=(\one-k!P_{1^3})\ot {P_{1^3}}^{\ot n-1}$ taken over $k$ copies of Haar-random states
shared among $n$ inhomogeneous systems, $\ket{\psi}\in\C^d$ with $d=d_1\cdots d_n$. To have nonvanishing antisymmetrizers $P_{1^k}$ acting on ${\C^{d_i}}^{\ot k}$, we will assume $k\leq d_{\min}$ where $d_{\min}$ is the minimal local dimension. The expression to evaluate is
\begin{equation}
    \mathbb{E}[\langle\WW\rangle] = \int_{U\in\mathbf{U}(d)} \tr\left(\WW \big (U \ket{\psi}\bra{\psi}U^{\dag}\big )^{\ot k}\mathbf{d}_\mu(U)\right)\,.
\end{equation}
Due to Theorem 4.3 in~\cite{KuengHaar_2019}, it holds that
\begin{equation}
\int_{U\in\mathbf{U}(d)} \big (U \ket{\psi}\bra{\psi}U^{\dag}\big )^{\ot k}\mathbf{d}_\mu(U) = \binom{d+k-1}{k}^{-1}\frac{1}{k!}\sum_{\sigma_d\in S_k} \sigma_d\,,
\end{equation}
where each permutation $\sigma_d=\sigma_{d_1}\ot\cdots\ot\sigma_{d_n}$ permutes $k$ tensor factors of dimension $d=d_1\cdots d_n$. Then
\begin{equation}\label{eq:EhaarBin}
\mathbb{E}[\langle\WW\rangle]\binom{d+k-1}{k} = \frac{1}{k!} \sum_{\sigma\in S_k}\tr(\sigma_{d_1})\prod_{i=2}^n\tr(P_{1^k}\sigma_{d_i}) - k!\prod_{i=1}^n\tr(P_{1^k}\sigma_{d_i})\,.
\end{equation}
Now we need to make use of the following property of the antisymmetrizer,
\begin{equation}\label{eq:PropAntiPerm}
    P_{1^k}\sigma = (-1)^{\sign(\sigma)}P_{1^k}\,.
\end{equation}
To see this, first note that (1) the product of two permutations $\sigma,\pi\in S_k$ is again a permutation $\tau\in S_k$ with $\sign(\tau)=\sign(\sigma)\sign(\pi)$, and (2) given $\sigma,\pi,\pi'\in S_k$ with $\pi\neq\pi'$, one has $\sigma\pi\neq\sigma\pi'$. Now we can insert $1=(-1)^{\sign(\sigma)}(-1)^{\sign(\sigma)}$ to Eq.~\eqref{eq:PropAntiPerm} and obtain
\begin{align}
    \sigma P_{1^k}&=\frac{1}{k!}\sum_{\pi\in S_k}(-1)^{\sign(\pi)}(-1)^{\sign(\sigma)}(-1)^{\sign(\sigma)}\sigma\pi \\
    &= \frac{1}{k!}(-1)^{\sign(\sigma)}\sum_{\tau\in S_k}(-1)^{\sign(\tau)}\tau \\
    &= (-1)^{\sign(\sigma)}P_{1^k}\,.
\end{align}
With this one can reduce Eq~\eqref{eq:EhaarBin} to
\begin{equation}
\mathbb{E}[\langle\WW\rangle]\binom{d+k-1}{k} = \frac{1}{k!}\sum_{\sigma\in S_k} \tr(\sigma)\tr(P_{1^k})^{n-1}\Big ((-1)^{\sign(\sigma)}\Big )^{n-1} - k!\sum_{\sigma\in S_k} \tr(P_{1^k})^{n}\Big ((-1)^{\sign(\sigma)}\Big )^{n}\,.
\end{equation}
Now we need to distinguish between $n$ being even and odd. If $n$ is even, then we have
\begin{align}
\mathbb{E}[\langle\WW\rangle] &\propto \frac{1}{k!}\tr(P_{1^k})^{n-1}\bigg ( \sum_{\sigma\in S_k}\tr(\sigma)(-1)^{\sign(\sigma)} - k!^2\tr(P_{1^k}) \bigg ) \\
&= \tr(P_{1^k})^n(1-k!)<0\,,
\end{align}
since by definition $\sum_{\sigma\in S_k}\tr(\sigma)(-1)^{\sign(\sigma)}=k!\tr(P_{1^k})$ holds. If $n$ is odd, one can similarly show that $\mathbb{E}[\langle\WW\rangle] = 0$.

\section{Sampling Werner states in the group ring}\label{app:WGFormalism}
We explain how we sample random Werner states in Fig.~\ref{fig:WernerDetProjections} following the lines of Ref.~\cite{Huber2022DimFree}.
Let $\sigma\in S_n$ and let $\eta_d(\sigma)$ be its representation permuting the $n$ tensor factors of $(\mathbb{\C}^d)^{\ot n}$,
\begin{equation}
\eta_d(\sigma)\ket{i_1} \ot \dots\ot \ket{i_n}
=
\ket{i_{\sigma^{-1}(1)}} \ot \dots \ot \ket{i_{\sigma^{-1}(n)}}\,.
\end{equation}
For each element $a=\sum_{\sigma\in S_n}\alpha_\sigma\sigma \in \C S_n$, we denote its representation in the Hilbert space as $\eta_d(a)$. For each partition of $n$ labeled by $\lambda$, define the centrally primitive idempotent
\begin{equation}
\omega_\lambda = \frac{\chi_\lambda(\id)}{n!}\sum_{\sigma\in S_n}\chi_\lambda(\sigma)\sigma^{-1}\,
\end{equation}
and its representation in the Hilbert space $P_\lambda=\eta_d(\omega_\lambda)$, where $\chi_\lambda(\id)$ is the character of the identity in the irreducible representation labeled by $\lambda$. For example, $P_{1^n}$ is the projector onto the antisymmetric subspace of $(\C^d)^{\ot n}$. 

We now define a trace $\tau$ on the elements of $S_n$ as
\begin{equation}
\tau(\sigma)=
\begin{cases}
n!\quad\text{if $\sigma=\id$} \\
0\quad\text{else}
\end{cases},
\end{equation}
and the Weingarten operator $\Wg(d,n) = \eta_d(\wg(d,n))$ where~\cite{Collins2006IntHaar,Procesi2020noteWg}
\begin{equation}\label{eq:defwg}
    \wg(d,n)=\frac{1}{n!}\sum_{\substack{\lambda \vdash n \\ |\lambda|\leq d}}\frac{\tau(\omega_\lambda)}{\tr(P_\lambda)}\omega_\lambda \,.
\end{equation}
We denote $J_d=\ker(\eta_d)^\perp$ the subspace of $\mathbb{\C}S_n$ with nonzero representation $\eta_d$ acting on $(\mathbb{\C}^d)^{\ot n}$,
\begin{equation}
J_d = \sum_{\substack{\lambda \vdash n \\ |\lambda|\leq d}} \omega_\lambda \C S_n\,,
\end{equation}
which is spanned by $\omega_\lambda$ such that the height of the Young tableau $\lambda$ is not larger than $d$. Now we are ready for the following Lemma.
\begin{lemma}[Lemma 1 in \cite{Huber2022DimFree}]
\label{lem:Trace-Tau}
\item (1) For all $r\in J_d$ and $b\in\C S_n$ it holds that
\begin{equation}
n!\tr\big (\Wg(d,n)\eta_d(r)\eta_d(b)\big ) = \tau(rb)\,.
\end{equation}
\item (2) Let $r\in J_d$. Then $r$ is a state if and only if $n!\Wg(d,n)\eta_d(r)$ is a state in $L((\C^d)^{\ot n})$.
\end{lemma}
The first part of Lemma~\ref{lem:Trace-Tau} allows us to compute inner products of operators which are representations of $\C S_n$ in the group ring. In order to evaluate expressions of the form~\eqref{eq:non-lin_precise}, one needs to compute tensor products of states and witnesses. For that, we will need the following variation of the first part of Lemma~\ref{lem:Trace-Tau}.
\begin{lemma}
For all $a\in J_d$ and $b\in\C S_n$ it holds that
\begin{equation}\label{eq:Trace-TauRestated}
\tr\big (\eta_d(a)\eta_d(b)\big ) = \frac{\tau(\wg(d,n)^{-1}ab)}{n!}\,.
\end{equation} 
\end{lemma}
\begin{proof}
First note that the inverse of the Weingarten operator in $\C S_n$ can be written as
\begin{equation}
    \wg(d,n)^{-1} = n!\sum_{\substack{\lambda \vdash n \\ |\lambda|\leq d}}\frac{\tr(P_\lambda)}{\tau(\omega_\lambda)}\omega_\lambda
\end{equation}
since the sum of any subset of isotypic components $\omega_\lambda$ is the identity in the subspace they have support on. Now let $r=\wg(d,n)a$ and assume $a\in J_d$. Lemma~\ref{lem:Trace-Tau} applies because $r\in J_d$ as well, and thus Eq.~\eqref{eq:Trace-TauRestated} follows from Eq.~\eqref{eq:Trace-Tau}.
\end{proof}
Now we are able to compute the expectation value of $n$ tensor products of operators $W_i=\eta_d(\omega_i)$ with respect to $k$ copies of a quantum state $\varrho=n!\Wg(d,n)\eta_d(r)$, using that the representation $\eta_d$ is linear. We have
\begin{equation}\label{eq:trTauTensProd}
\begin{aligned}
    \tr(\varrho^{\ot k}W_1\ot...\ot W_n) &= 
    \tr\big (n!^k \Wg(d,n)^{\ot k}\eta_d(r)^{\ot k}\eta_d(w_1)\ot...\ot\eta_d(w_n)\big ) \\
    &= n!^k\tr\big (\eta_d(\wg(d,n)^{\times k}r^{\times k})\eta_d(w_1\times...\times w_n)\big ) \\
    &= \frac{n!^k}{(kn)!}\tau\big (\wg(d,kn)^{-1}\wg(d,n)^{\times k}r^{\times k}w_1\times...\times w_n\big )\,.
\end{aligned}
\end{equation}
This gives us the following practical recipe to sample Werner states. First we fix an operator $\WW$ of the form~\eqref{eq:NonlinearWitness} which is a tensor product of witnesses and positive operators $W_i=\eta_d(w_i)$ with $w_i\in\C S_k$. Then we compute iteratively its expectation values with respect to $k$-fold copies of $n$-partite Werner states, according to Eq.~\eqref{eq:trTauTensProd}. At each iteration, we take $s=\sum_\sigma s_\sigma\sigma\in\C S_n$ choosing the coefficients $s_\sigma \in \C$ randomly, and compute
\begin{equation}\label{eq:appRandomNormalizedState}
r=\frac{ss^\dag}{\tau(ss^\dag)}
\end{equation}
with $s^\dag=\sum_\sigma s_\sigma^*\sigma^{-1}$, such that $r$ is a state in $\C S_n$ and therefore $\varrho = n!\Wg(d,n)\eta(r)$ is a quantum state in the Hilbert space $\cdn$. Finally we compute the expectation value $\langle\WW\rangle_{\varrho^{\ot k}}$ via Eq.~\eqref{eq:trTauTensProd}.
\end{widetext}
\end{document}